\documentclass[journal, onecolumn]{IEEEtran}
\usepackage{amsmath}
\usepackage{amsthm}
\usepackage{amsfonts}
\usepackage{amssymb}
\usepackage{url}
\usepackage{cite}
\usepackage{graphicx}

\newtheorem{proposition}{Proposition}
\newtheorem{lemma}{Lemma}
\begin{document}
\title{Throughput Maximization for Wireless Powered Communications Harvesting from \\Non-dedicated Sources\thanks{Hongxing Xia, Yongzhao Li and Hailin Zhang are with State Key Laboratory of Integrated Services Networks, Xidian University, Xi'an, 710071,China.}}

\author{Hongxing Xia,  Yongzhao Li and Hailin Zhang}

\maketitle

\begin{abstract}
We consider the wireless powered communications where users harvest energy from non-dedicated sources. The user follows a harvest-then-transmit protocol:  in first phase of a slot time the source node  harvests energy from a nearby conventional Access Point,  then transmit information to its destination node or relay node in the second phase. We obtain the optimal\textit{ harvesting ratio} to maximize the expected throughput for direct transmission (DT )and decode forward (DF) relay under outage constraint, respectively. Our results reveal that the optimal harvest ratio for DT is dominated by the outage constraint while for DF relay, by the data causality .
\end{abstract}

\begin{IEEEkeywords}
Wireless powered communications, throughput optimization,
harvest ratio, non-dedicated sources.
\end{IEEEkeywords}

\section{Introduction}
By extracting radio frequency(RF) energy from the nearby wireless access points(APs), the energy constrained wireless devices like wireless sensors can theoretically  posses perpetual lifetime  \cite{Visser2013Rf}. Comparing with traditional energy harvesting techniques (e.g. solar energy, thermal gradients energy  and vibrations or movements energy),  RF energy harvesting technique has advantages such as all-weather operation and a compact harvester\cite{paradiso2005energy}. Very recently, a commercial program named FreeVolt has been launched in London, devoting to absorb energy from cell towers, Wi-Fi access points and TV broadcasters to charge low energy IoT devices \cite{FreeVolt15}. As the feasibility of RF-energy harvesting technique improves, the wireless powered communications (WPC) has attracted growing attentions \cite{huang2015some}.

For wireless nodes without fixed power supply,  a harvest-then-transmit protocol has been commonly adopted where the node first harvest energy from radio signal transmitter and then use the harvested energy to transfer information to its intended receiver \cite{yin2013throughput,ju2014throughput,Rui2014Throughput,Yue2015Spatial}. The half-duplex operation mode produces an interesting trade-off for energy harvesting communications: what is the optimal harvesting-time to transmitting-time ratio for obtaining maximum achievable throughput. In \cite{yin2013throughput}, the impact of energy harvesting rate on throughput optimization was studied. Their results indicated that the optimal harvest-ratio decreases with increasing energy harvest rate.  In \cite{ju2014throughput}, the doubly near-far problem was investigated for  WPC networks consisting many RF-powered nodes. To counter the unfair throughput allocation among the near and far users, a common-throughput optimization problem was formulated and solved. In similar settings, \cite{Rui2014Throughput} studied the sum-throughput maximization problem for the case that the nodes can save energy for later use. A large-scale WPC network was studied in \cite{Yue2015Spatial}, where the node's spatial throughput is maximized subject to successful information transmission probability constraint.

As the energy harvested from RF is  quite low and then one can use relay node to improve the transmission rate of the source's information \cite{ahmed2007throughput}.  The outage probability of energy harvesting relay-aided link over fading channel was studied in \cite{li2016outage}. However, the  harvesting profile of RF energy was not considered in the paper.

In previous works, people usually assume that  the wireless nodes harvest energy from dedicated sources, where the AP and the nodes can operate synchronously in a cooperative mode. In this context,  the AP keeps silent when the nodes are transmitting information. However, harvesting from dedicated sources like Hybrid AP  is still not practical at resent because of high upgrade cost and low energy efficiency. On the contrary, harvesting from non-dedicated sources like WiFi, small base stations and TV stations is more practical and easy to implement.

In this paper, we consider the throughput optimization problem in the WPCs powered by non-dedicated sources. In particular, we assume that the wireless powered user harvests energy from a conventional nearby AP, and the AP keeps transmitting to its associate receiver even when the user starts transmitting information. Therefore, from the view of wireless powered user, the AP first act as an energy source and then as an interference source.
To the best of my knowledge, throughput optimization under this settings has not been addressed as so far. Interestingly, since the energy and interference come from the same source, improving the transmitting power of AP will not bring any profit to WPC users. And then the {\it harvest ratio} is becoming  the only parameter that affects the performance.

Since the harvested RF energy is usually very weak, the outage may occur more frequently than conventional communications. As well known, the outage probability decreases with increasing {\it signal-to-interference} ratio (SIR).  However, as will be shown later, the throughput is a quasi-concave function over SIR. Then there must be a trade-off between the throughput and the outage probability. Different from previous work, We integrate the outage probability constraint to the throughput optimization problem in this paper.Besides,  the throughput of the WPC with relay has not been extensively studied as far as I know. And we will examine this problem considering outage as well as data causality constraint.

The main contributions of this paper are listed as follows.
\begin{itemize}
\item We propose a novel non-dedicated sources powered wireless communication model.  The protocol of direct transmission and decode-and-forward relay transmission are presented, respectively.
\item The maximization of expected throughput for direct transmission  subject to outage probability constraint is formulated and solved.  The  upper bound of the expected  throughput is given in close form.
\item The maximization of expected throughput for DF relay transmission subject to outage and data causality constraints is formulated. We solve this problem by dividing it into two sub-problems.
\item Our results show that the optimal expected throughput for direct transmission is dominated by the outage constraint in most practical scenarios , while for DF relay transmission by the data causality constraint.
\end{itemize}

\section{Preliminaries}
\subsection{System model}

In this paper, we consider a wireless powered communication as shown in Fig.\ref{Fig:System}. We assume AP is in full-load operation and transmits with fixed power $P_{A}$.
The energy harvesting nodes $S, R$ has no fixed power supply and extract energy from radio signal radiated by a conventional AP. Suppose the AP and energy harvesting nodes operate in the same frequency band.

We consider two WPC schemes:  direct transmission (DT) and Decode and Forward(DF) relay. For the case without relay, the system follows a {\it harvest-then-transmit} MAC protocol as shown in Fig.\ref{Fig:protocol}(a): In each time slot, the source node $S$ harvests energy from the AP in the first $\alpha T$ , while employ the remaining time fraction of $(1-\alpha)T$ to directly transmit information to the destination node $D$. The symbol $\alpha$ denotes {\it harvest ratio} and $0<\alpha<1$ .  We further assume that node $S$ uses up all the harvested energy to transmit information in the second phase.

For the case with relay,  we assume there is a relay node $R$ helps node $S$ transfer information to the destination. The system follows a {\it harvest-transmit-relay} MAC protocol as shown in Fig.\ref{Fig:protocol}(b): The source node $S$ and the relay node $R$ first harvest energy from the AP for $\alpha T$ time. Secondly,  node $S$ transfer information to node $R$ with all of the harvested energy in $\beta T$ time, $0<\beta<1$. Lastly,  node $R$ relay the information from $S$ to $D$ in $(1-\alpha-\beta)T$ time, by using up the harvested energy in the frist phase. To simplify analysis, we assume normalized slot duration, i.e. $T=1$,  in the left part of this paper.

\begin{figure}
\centering
    \includegraphics[width=.5\linewidth]{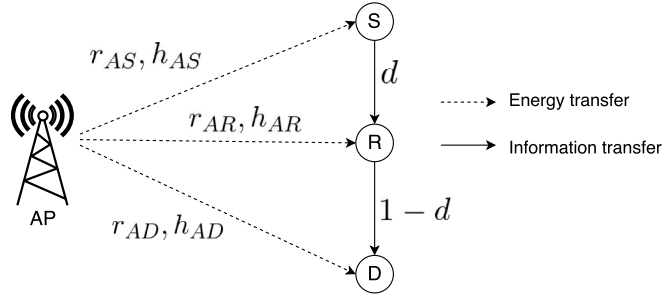}
    \caption{System model}
    \label{Fig:System}
\end{figure}

\begin{figure}
\centering
    \includegraphics[width=0.5\linewidth]{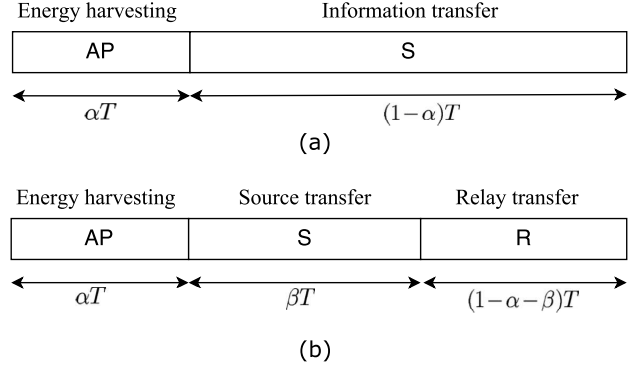}
    \caption{MAC protocol for direct transmission and DF relay transmission}
    \label{Fig:protocol}
\end{figure}

\textbf {Non-dedicated Sources Assumption}:  Other than previous works, we assume energy harvesting nodes have no cooperation with the AP, which means the AP keeps transmitting to serve its own users no-matter the node $S$ is transmitting or not. For such assumption, the AP first acts as an energy source and latter as an interferer. Some new challenges should be observed and formulated.  This assumption is real for those low-power sensors harvesting energy from nearby information AP, like WiFi and small cell base stations.

\subsection{Energy harvesting model}

We assume that the energy transfer channels from AP to $S, R$ and $D$ are subject to Rayleigh fading with unit mean and large-scale path loss.  Let $r_{AS}$ , $r_{AR}$ and $r_{AD}$ denote the distance from AP to $S$, $R$ and $D$, respectively.  Then in the energy harvesting phase, the accumulated energy of the node $S$ and $R$ are:
\begin{equation}\label{eqn:es}
E_{S}=\alpha \zeta P_A h_{AS} r_{AS}^{-\mu}
\end{equation}
and
\begin{equation}\label{eqn:er}
E_{R}=\alpha \zeta P_A h_{AR} r_{AR}^{-\mu},
\end{equation}
respectively, where $0<\zeta<1$ is the energy harvesting efficiency, $ h_{AS}, h_{AR}$ are independent and identically distributed (i.i.d.) exponential random variables with unit mean, $\mu>2$ is the path-loss exponent. For simplification, we assume $\zeta=1$ in the remaining parts.

\subsection{Information transfer model}

In transmission without relay, only the direct transmission between $S$ and $D$ is available. Without loss of generality, the distance between them $r_{SD}$ is set to be $1$. The channel power gain  between $S$ and $D$ is assumed to be only determined by their distance as: $h_{SD}=r_{SD}^{-\mu}=1$. The channel power gain between the AP and the node $D$ is given by $h_{AD}r_{AD}^{-/mu}$. Let $x_A(t)$ and $x_S(t)$ with zero mean and unit power, denote the transmit signal of the AP and the node $S$, respectively. For direct transmission link, the baseband equivalent model for this channel is
\begin{equation}
y_{SD}(t)=\sqrt{P_S}x_S(t)+\sqrt{P_A h_{AD}r_{AD}^{-\mu}}x_A(t)+n_D(t)
\end{equation}
where $y_{SD}(t)$ denotes the received signal from direct link, $n_D(t)$ is the additive white noise with power $\sigma_{n_D}^2$. As we assume the node $S$ use up the harvested energy to transfer information, and the previous research showed that keeping constant transmitting power for such energy harvesting system will achieve the  maximum channel capacity \cite{devillers2012general}. Therefore  the transmitting power of the node $S$  is,
\begin{equation}\label{eqn:ps}
P_S=\frac{E_{S}}{1-\alpha}
\end{equation}

The receiving SINR of node $D$ for direct transmission is
\begin{equation}\label{eqn:sinr-dt}
\gamma_{DT}=\frac{P_S}{P_Ah_{AD}r_{AD}^{-\mu}+\sigma_{n_D}^2},
\end{equation}
The throughput of direct link $S-D$ is given as following,
\begin{equation}\label{eqn:th-dt}
R_{DT}=(1-\alpha)\log(1+\gamma_{DT}).
\end{equation}

For DF relay transmission, we assume that the distance between the source node and the relay node is $d$, the distance between relay and destination $1-d$ \cite{ahmed2007throughput}. Then the channel power gain of link $S-R$ and $R-D$ is, respectively, $d^{-\mu}$ and $(1-d)^{-\mu}$. The information transmission phase is divided into two sub-phases. In the first $\beta $ time, the information is transmitted to the relay node $R$, the received signal at $R$ is
\begin{equation}\label{eqn:ysr}
y_{SR}(t)=\sqrt{P_S^{co}d^{-\mu}}x_S(t)+\sqrt{P_Ah_{AR}r_{AR}^{-\mu}}x_A(t)+n_R(t),
\end{equation}
where $n_R(t)$ denotes the white noise at node $R$. Accordingly, the received SINR of  node $R$ and  the throughput of $S-R$ link is given as (\ref{eqn:sinr-sr}) and (\ref{eqn:th-sr}) , respectively.
\begin{equation}\label{eqn:sinr-sr}
\gamma_{SR}=\frac{P_S^{co}d^{-\mu}}{P_Ah_{AR}r_{AR}^{-\mu}+\sigma_{n_R}^2},
\end{equation}
\begin{equation}\label{eqn:th-sr}
R_{SR}=\beta\log(1+\gamma_{SR}).
\end{equation}
Since we assume all the energy harvesting nodes use up the energy in their batteries or capacities,  the transmitting power of node $S$ in the cooperative mode is $P_S^{co}=\frac{E_S}{\beta}$.

In the following $1-\alpha-\beta$ time, the relay node transfer the information to the destination node $D$ by using all of its harvested energy. The received baseband-equivalent signal at node $D$ is
\begin{equation}\label{eqn:yrd}
y_{RD}=\sqrt{P_R (1-d)^{-\mu}}x_R(t)+\sqrt{P_Ah_{AD}r_{AD}^{-\mu}}x_A(t)+n_D(t),
\end{equation}
where the transmitting power of node $R$ is $P_R=\frac{E_R}{1-\alpha-\beta}$ , $x_r(t)$ is the relay signal with unit mean power, $n_D(t)$ denotes the noise signal with power $\sigma_{n_D}^2$. Similar to (\ref{eqn:sinr-sr}) and (\ref{eqn:th-sr}), the received SINR of node $D$ and throughput of  $R-D$ link is, respectively,
\begin{equation}\label{eqn:sinr-rd}
\gamma_{RD}=\frac{P_R (1-d)^{-\mu}}{P_Ah_{AD}r_{AD}^{-\mu}+\sigma_{n_D}^2},
\end{equation}
and
\begin{equation}\label{eqn:th-rd}
R_{RD}=(1-\alpha-\beta)\log_2(1+\gamma_{RD}).
\end{equation}

\subsection{Preliminary Mathematical Results}
\begin{lemma}\label{lem:cdf}
Assume $H_1$ and $H_2$ are independent exponential distribution variables with unit mean,  $k\in R^+$. Then for variable  $X=k\frac{H_1}{H_2}$,  the  probability density function (PDF)  is $\frac{k}{(k+x)^2}$, the cumulative distribution function (CDF) is $\frac{k}{k+x}$.
\end{lemma}
\begin{proof}
	See appendix \ref{app:1}
\end{proof}
\begin{lemma}\label{lem:mean}
Assume X is a random with PDF  $\frac{k}{(k+x)^2}$, $k\in R^+$, the expectation of function $f(X)=\log_2(1+X)$ is $\mathbb{E}[f(X)]=\frac{\log_2(1/k)}{1/k-1}$.
\end{lemma}
\begin{proof}
	See appendix  \ref{app:2}
\end{proof}

\section{Problem Formulation}\label{sec:probform}

From previous analysis we find that the system throughput and the received SINR  are both random variables for given time allocation strategy. In a specific time slot, the time allocation ratio can affect the throughput performance. For example, in direct transmission case, allocating more time to harvest energy will increase the transmitting power according to (\ref{eqn:ps}). Thus the received SINR would be  increased and lower outage probability can be achieved. However, this strategy would shorten the transmission duration thus decrease the system throughput.

To observe how the{ \it harvest-ratio } $\alpha$ affects the throughput as well as indirectly,  how the throughput varies with increasing SIR, we consider a deterministic case that the channel gains $h_{AS}$ and $h_{AR}$ are both assumed to equal to $1$. By introducing (\ref{eqn:es}), (\ref{eqn:ps}) to (\ref{eqn:sinr-dt}), we get the SINR of node $D$ as
\begin{equation}\label{eqn:gkk}
\gamma_{DT}=\frac{\alpha r_{AS}^{-\mu}}{(1-\alpha)r_{AD}^{-\mu}+\sigma_{n_D}^2}.
\end{equation}
Considering the RF energy powered communications are usually low-power and communication range limited, the distance between two nodes is far less than that between the AP and the nodes. Therefore, we assume $r_{AS}\sim r_{AD}$. Besides, as noise power is far less than the signal radiated from the energy source, we further remove $\sigma_{n_D}^2$ from (\ref{eqn:gkk}). Then the SINR is reduced to
\begin{equation}\label{eqn:gkk1}
\gamma_{DT}=\frac{\alpha }{1-\alpha}.
\end{equation}

The throughput of direct transmission link can be expressed as
\begin{displaymath}\label{eqn:rdt}
R_{DT}=(1-\alpha)\log_2{\frac{1}{1-\alpha}}.
\end{displaymath}

If we only aim to maximize the throughput, the optimal { \it harvest-ratio }  $\alpha$ can be easily derived as $1-1/e$ from theorem 1 in \cite{yin2013throughput} . However, if we plot a curve  of throughput over SIR, as shown in Fig.\ref{Fig:Th1}, we can find that the trade-off exists between throughput maximization and QoS optimization. The trade-off comes from the fact that  allocating more time to harvest energy will always increase the SIR from (\ref{eqn:gkk1}), but it just increases the throughput before the critical point $e-1$ while decreases after that point. As well known, to avoid the occurrence of outage, the SIR should exceed certain threshold $\gamma_{o}$. Therefore, the optimal  { \it harvest-ratio } $\alpha$ integrating the outage constraint should be:
\begin{displaymath}
{\alpha}^* = \left\{ \begin{array}{ll}
\gamma_o & \gamma_o\ge e-1\\
1-1/e & \gamma_o< e-1
\end{array} \right.
\end{displaymath}

\begin{figure}
\centering
    \includegraphics[width=0.5\linewidth]{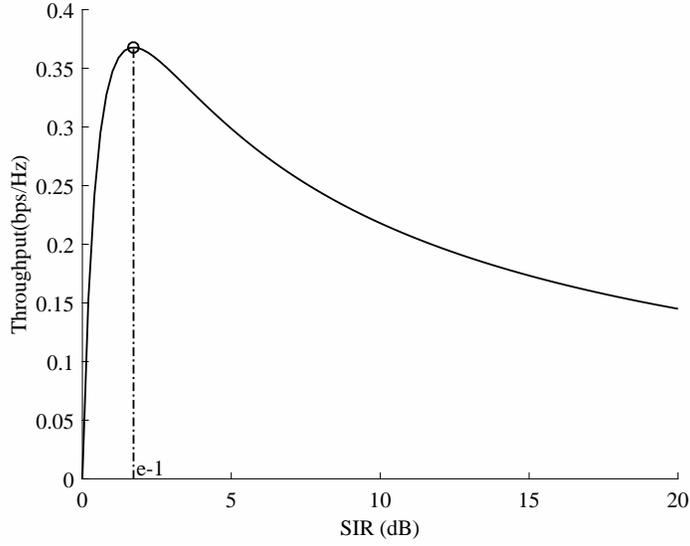}
    \caption{Throughput versus SIR for deterministic channel.}
    \label{Fig:Th1}
\end{figure}

In this paper, we focus on the stochastic case where the channel power gain follows exponential distribution with unit mean. Our goal is to maximize the long-term expectation of the achievable throughput  subject to outage probability constraint. For direct transmission protocol we  formulate the following optimization problem,

\begin{eqnarray}\label{eqn:p1}
\mathrm{\mathbf{P1}}:\underset{\mathit{\alpha}}{\mathbf{max}} &  & \mathbb{E} [R_{DT}] \label{eq:p3} \\
\mathrm{\mathbf{s.t.}} &  & P_{\gamma_o, DT}^{out}\le\theta\\
&&0<\alpha<1,
\end{eqnarray}
where $\mathbb{E}[R_{DT}]$ denotes the expectation of $R_{DT}$,  $\theta$ the maximum tolerable outage probability and $P_{DT}^{out}$ the outage probability with SINR threshold $\gamma_o$.

For DF relay transmission, we formulate the following optimization problem,
\begin{eqnarray}\label{eqn:p2}
\mathrm{\mathbf{P2}}:\underset{\mathit{\alpha, \beta}}{\mathbf{max}} &  & \mathbb{E} [R_{DF}] \label{eq:p3} \\
\mathrm{\mathbf{s.t.}} &  & P_{\gamma_o, DF}^{out}\le\theta\\
&&\mathbb{E}[R_{SR}]\le\mathbb{E}[R_{RD}]\label{eqn:p2-3}\\
&&0<\alpha<1\\
&&0<\beta<1\\
&&\alpha+\beta<1.
\end{eqnarray}
The constraint (\ref{eqn:p2-3}) is the  data causality constraint that the average throughput of $S-R$ link can not be larger than that of $R-D$ link \cite{orhan2015energy}.

Note that for interference limited channel assumption, i.e. $\sigma_{n_D}^2=0$, the SIR and the throughput is irrelevant to the transmitting power of the AP. That is to say increasing the AP's transmitting power will not lead to increasing of achievable throughput. The only parameters that determine the system performance is the time allocation ratio $\alpha$ and $\beta$, which are to be optimized in this paper.

\section{Throughput Maximization for Direct Transmission}\label{sec:dt}

In this section, we first get the outage probability and the average throughput for direct transmission based on the distribution of SIR. Then we solve the optimization problem by convex optimization technique.

\subsection{Outage Probability and Expected Throughput}
Similar to assumption in section \ref{sec:probform}, we get the received SIR of node $D$ by combining (\ref{eqn:es}), (\ref{eqn:ps}) and (\ref{eqn:sinr-dt}) as following,
\begin{equation}\label{eqn:gdt}
\gamma_{DT}=\frac{\alpha}{1-\alpha}\cdot\frac{h_{AS}}{h_{AD}}.
\end{equation}
Since the channel gain $h_{AS}$ and $h_{AD}$ are exponential random variables with unit mean, we get the PDF of $\gamma_{DT}$ as following, according to lemma \ref{lem:cdf}.
\begin{equation}\label{eqn:pdfgdt}
f_{\Gamma_{DT}}(\gamma)= \dfrac{\alpha(1-\alpha)}{[(1-\alpha)\gamma+\alpha]^2}.
\end{equation}

For  given outage threshold $\gamma_o$, the outage probability is thus given by
\begin{align}
P_{\gamma_o,DT}^{out} &= \mathbb{P}(\gamma_{DT}\le \gamma_o) \notag \\
&=\int_{0}^{\gamma_o}f_{\Gamma_{DT}}(\gamma)d\gamma \notag\\
&=\frac{(1-\alpha)\gamma_o}{\alpha +(1-\alpha)\gamma_o}.\label{eqn:pout}
\end{align}

Based on distribution of the SIR, we get the expected throughput of direct transmission over numerous time slots as following,
\begin{align}\label{eqn:erdt}
\mathbb{E}[R_{DT}]=&\mathbb{E}[(1-\alpha)\log_2(1+\gamma_{DT})] \\
=&(1-\alpha)\int_{0}^{\infty}\log_2(1+\gamma_{DT})f_{\Gamma_{DT}}(\gamma)d\gamma\notag\\
\overset{(a)}{=}&\frac{\alpha(1-\alpha)}{1-2\alpha}\log_2(\alpha^{-1}-1)\notag,
\end{align}
where (a) comes by replacing $k$ with $\frac{\alpha}{1-\alpha}$ in Lemma \ref{lem:mean}. It is not difficult to prove that $\mathbb{E}[R_{DT}]$ is a continuous function for $\alpha\in(0,1)$. Substitute (\ref{eqn:pout}) and (\ref{eqn:erdt}) into problem P1, and make some simplifications we get the following equivalent problem:
\begin{eqnarray}
\mathrm{\mathbf{P3}}:\underset{\mathit{\alpha}}{\mathbf{max}} &  &\mathbb{E}[R_{DT}] \label{eqn:obj2}\\
\mathrm{\mathbf{s.t.}} &  &\alpha \ge \frac{\gamma_o(1-\theta)}{\theta +\gamma_o (1-\theta)} \label{eqn:cond2}\\
&&o<\alpha<1.
\end{eqnarray}

\subsection{Solution of Problem P3}
In this part, we first prove the objection function of problem P3 is convex and then obtain the solution by using convex optimization technique.
\begin{lemma}\label{lem:concdt}
The objective function in P3 is concave.
\end{lemma}
\begin{proof}
	See Appendix for details.

\end{proof}
\begin{proposition}\label{pro:opt}
	The optimal $\alpha$ for problem P3 is ${\alpha}^*=\max \Big(0.5,\big(\frac{\theta}{(1-\theta)\gamma_0}+1\big)^{-1}\Big)$.
\end{proposition}

\begin{proof}
	As problem P3 is an univariate maximization problem with bounded constraint, we could simply maximize the objection function and compare the optimal parameter ${\alpha}^*$ with the bound to get the actual optimal value. We get ${\alpha}^*$ by letting the first derivative of $\mathbb{E}[R_{DT}]$ equal zero,
	\begin{align}
	\frac{d\mathbb{E}[R_{DT}]}{d{\alpha}}&=\frac{{\alpha}^2+(\alpha-1)^2}{(2\alpha-1)^2}\log_2(\frac{\alpha}{1-\alpha})-\frac{1}{(2\alpha-1)\ln 2}=0.\notag
	\end{align}
	This equation can be equivalently translated to the following one
	\begin{equation}
	\ln(\frac{\alpha}{1-\alpha})=\frac{2\alpha-1}{{\alpha}^2+(\alpha-1)^2}    \label{eqn:diff1}
	\end{equation}
	The exact solution for (\ref{eqn:diff1}) is $\alpha^*=0.5$. Considering the lower bound of $\alpha$ subject to (\ref{eqn:cond2}), we can intuitively get Proposition \ref{pro:opt}.
\end{proof}

From Proposition \ref{pro:opt} we find that the expected  throughput would be no more than $0.5\log_2 e\approx 0.7213$ bps/Hz even if we remove the outage constraint, which implies that the spectral efficiency of the non-cooperative WPCs is quite low. However, we can find broad applications  for this  protocol in low-rate wireless sensors networks.

\subsection{Simulation Result}
In Fig.\ref{Fig:r_k}, we present the simulation of average achievable throughput to verify our analytical results. The simulation results are obtained by averaging over 10,000 independent Rayleigh  channel realizations. The analytical results is plotted according to (\ref{eqn:erdt}). We find that the two curves matches well and then the analysis framework is verified.
\begin{figure}
\centering
    \includegraphics[width=0.5\linewidth]{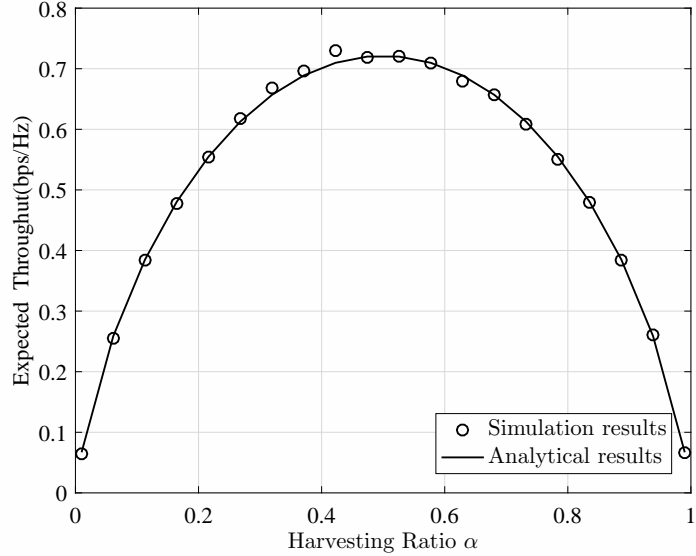}
    \caption{Average achievable throughput versus harvesting ratio.}
    \label{Fig:r_k}
\end{figure}

In Fig.\ref{Fig:r_zeta}, the optimal average throughput versus energy harvesting efficiency $\zeta$ under outage probability threshold 0.05 and 0.02 are depicted, respectively. The outage SIR threshold is set as $\gamma_o=-13$ dB. The curve without outage probability constraint is also given for comparison. We can see that 1) the throughput performance will be improved with weaker outage probability constraint as expected, and be maximized if the outage probability constraint is completely relaxed; 2)The upper bound given by Proposition \ref{pro:opt} is testified  at point (1, 0.73).
\begin{figure}
\centering
    \includegraphics[width=0.5\linewidth]{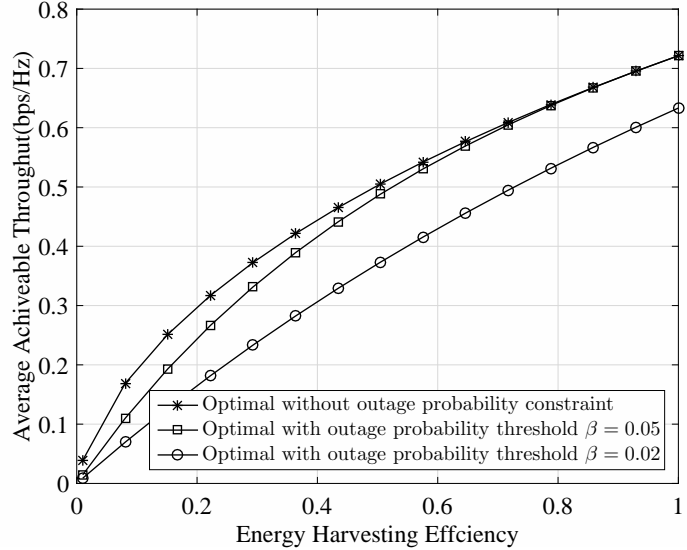}
    \caption{Average achievable throughput versus energy harvesting efficiency.}
    \label{Fig:r_zeta}
\end{figure}

To investigate how the varying outage SIR threshold impacts the optimal {\it harvest-ratio}, we plot Fig.\ref{Fig:k_gamma}. The simulation parameters are the same as in Fig.\ref{Fig:r_k}. We present the cases with higher and lower outage probability threshold, respectively. Our results imply that 1) the optimal {\it harvest-ratio} is determined by the outage probability constraint for most practical scenarios, i.e. $\gamma_o>-10$ dB, while in the lower SIR regime it is just determined by the objective function itself; 2) The optimal $\alpha$ for the case with a stronger outage constraint ($\theta=0.02$) is always higher than that with a weaker one ($\theta=0.05$). This can be explained as that to meet stronger outage constraint, more time should be allocated to harvest energy and then high SIR is achieved; 3)The optimal $\alpha=0.5$ given by Proposition \ref{pro:opt}  is checked.
\begin{figure}
\centering
    \includegraphics[width=0.5\linewidth]{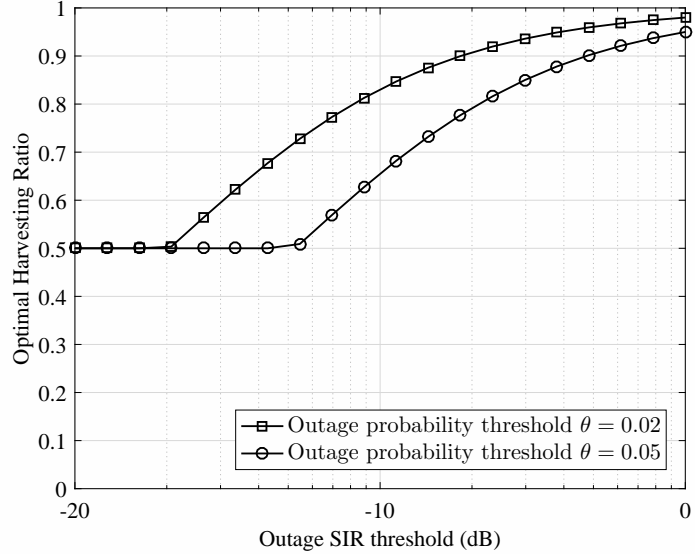}
    \caption{Optimal harvesting ratio versus outage SIR threshold.}
    \label{Fig:k_gamma}
\end{figure}
\section{Throughput maximization for DF relay transmission}
\subsection{Outage Probability and Expected Throughput}
For DF relay aided communication, the outage occurs when anyone of the two links,  $S-R$ or $R-D$ link,  is outaged. Thus the overall outage probability under DF cooperative protocol can be presented as
\begin{equation}\label{eqn:poutdf}
P_{\gamma_o, DF}^{out}=1-(1-P_{\gamma_o, SR}^{out})(1-P_{\gamma_o, RD}^{out})
\end{equation}
We next get the expression for $P_{\gamma_o, SR}^{out}$ and $P_{\gamma_o, RD}^{out}$. Following the same assumption as sections \ref{sec:dt}, we can get the received SIR at node $R$  from (\ref{eqn:sinr-sr}),  and the received SIR at node $D$ from (\ref{eqn:sinr-rd}) as following,
\begin{equation}\label{eqn:sir-sr}
\gamma_{SR}=\frac{\alpha d^{-\mu}}{\beta}\cdot\frac{h_{AS}}{h_{AR}},
\end{equation}

\begin{equation}\label{eqn:sir-rd}
\gamma_{RD}=\frac{\alpha d^{-\mu}}{1-\alpha-\beta}\cdot\frac{h_{AR}}{h_{AD}}.
\end{equation}

The distribution of $\gamma_{SR}$ and $\gamma_{RD}$ are similar to that of direct transmission protocol. We can obtain them by replacing $k$ with $\frac{\alpha d^{-\mu}}{\beta}$ and $\frac{\alpha(1- d)^{-\mu}}{1-\alpha-\beta}$ in Lemma \ref{lem:cdf}, respectively, as following.
\begin{equation}\label{eqn:sir-sr-dist}
f_{\Gamma_{SR}}(\gamma)= \dfrac{\alpha\beta d^{-\mu}}{[\alpha d^{-\mu}+\gamma\beta]^2},
\end{equation}
and
\begin{equation}\label{eqn:sir-rd-dist}
f_{\Gamma_{RD}}(\gamma)= \dfrac{\alpha(1-\alpha-\beta) (1-d)^{-\mu}}{[\alpha d^{-\mu}+\gamma(1-\alpha-\beta)]^2}.
\end{equation}

The outage probability of link $S-R$ is derived by introducing $\gamma_o$ to the CDF of SIR,
\begin{align}\label{eqn:pout-sr}
P_{\gamma_o,SR}^{out} &= \mathbb{P}(\gamma_{SR}\le \gamma_o) \notag \\
&=\int_{0}^{\gamma_o}f_{\Gamma_{SR}}(\gamma)d\gamma \notag\\
&=\frac{\beta\gamma_o}{\alpha d^{-\mu} +\beta\gamma_o}.
\end{align}
Similarly, the outage probability of link $R-D$ is derived as
\begin{align}\label{eqn:pout-rd}
P_{\gamma_o,RD}^{out} &= \mathbb{P}(\gamma_{RD}\le \gamma_o) \notag \\
&=\int_{0}^{\gamma_o}f_{\Gamma_{RD}}(\gamma)d\gamma \notag\\
&=\frac{(1-\alpha-\beta)\gamma_o}{\alpha(1-d)^{-\mu} +(1-\alpha-\beta)\gamma_o}.
\end{align}
Substitute (\ref{eqn:pout-sr}) and (\ref{eqn:pout-rd}) to (\ref{eqn:poutdf}),  we get the overall outage probability of the DF relay transmission as
\begin{equation}\label{eqn:poutdf2}
P_{\gamma_o, DF}^{out}=1-\dfrac{\alpha^2(d(1-d))^{-\mu}}{(\alpha d^{-\mu}+\beta\gamma_o)(\alpha(1-d)^{-\mu}+(1-\alpha-\beta)\gamma_o)}
\end{equation}

With the PDF of $\gamma_{SR}$ and $\gamma_{RD}$, the expected throughput of link $S_R$ and $R-D$ are derived as
\begin{align}\label{eqn:ersr}
\mathbb{E}[R_{SR}]=&\mathbb{E}[\beta\log_2(1+\gamma_{SR})] \\
=&\beta\int_{0}^{\infty}\log_2(1+\gamma_{SR})f_{\Gamma_{SR}}(\gamma)d\gamma\notag\\
=&\frac{\beta}{\beta d^{\mu}\alpha^{-1}-1}\log_2(\beta d^{\mu}\alpha^{-1})\notag
\end{align}
and
\begin{align}\label{eqn:errd}
\mathbb{E}[R_{RD}]=&\mathbb{E}[\beta\log_2(1+\gamma_{RD})] \\
=&(1-\alpha-\beta)\int_{0}^{\infty}\log_2(1+\gamma_{RD})f_{\Gamma_{RD}}(\gamma)d\gamma\notag\\
=&\frac{(1-\alpha-\beta)}{(1-\alpha-\beta)(1- d)^{\mu}\alpha^{-1}-1}\log_2((1-\alpha-\beta) (1-d)^{\mu}\alpha^{-1})\notag,
\end{align}
respectively.

For the DF relay transmission powered by RF energy, we assume that the direct link from $S$ to $D$ can be ignored due to its limited transmission ability. This assumption can also be used to get  the lower bound of an actual relay system.  Based on this assumption, the overall expected throughput is determined by the lower one between  $S-R$ and $R-D$ link,
\begin{equation}\label{eqn:erdt2}
\mathbb{E}[R_{DF}]=\min\{\mathbb{E}[R_{SR}],\mathbb{E}[E_{RD}]\}
\end{equation}
Considering that the data causality constraint  of problem P2, the objective function can be equally converted to
\begin{eqnarray}\label{eqn:p4}
\mathrm{\mathbf{P4}}:\underset{\mathit{\alpha, \beta}}{\mathbf{max}} &  & \mathbb{E} [R_{SR}]  \\
\mathrm{\mathbf{s.t.}} &  & P_{\gamma_o, DF}^{out}\le\theta \label{eqn:p4-1}\\
&&\mathbb{E}[R_{SR}]\le\mathbb{E}[R_{RD}] \label{eqn:p4-2}\\
&&0<\alpha<1\\
&&0<\beta<1\\
&&\alpha+\beta<1.
\end{eqnarray}

\subsection{Solution of Problem P4}
The objective function of P4 can be translated to finding the optimal \textit{harvest ratio} only considering the first two phases of the whole slot, which can be solved by using method similar to P3. Therefore, to simplify the analysis we divide the original problem two steps. In the first step, we will find the optimal \textit{harvest ratio} denoted by $\kappa=\frac{\alpha}{\beta}$. Note that we just find the optimal ratio of two times but not the times itself. In the second step, integrating the optimal ratio and the constraints in P4 we find the optimal harvesting time $\alpha$ and $S-R$ transmitting time $\beta$.

First, we see that finding the optimal \textit{harvest ratio } in P4 is not different from that in P2. The only difference is that in P3 the distance between $S-D$ is $1$ but in P4, the distance between $S-R$ is shorter and denoted as $d$. Setting $\alpha=\kappa\beta$, the corresponding expected throughput over fading power transfer channel is
\begin{equation}\label{eqn:ersr-p5}
\mathbb{E}[R_{SR}]=\frac{z}{(1+\kappa)}\frac{\log_2(\kappa^{-1}d^{\mu})}{\kappa^{-1}d^{\mu}-1},
\end{equation}
where $z=\alpha+\beta$ is assumed to be a constant in the first step. We name $z$ the \textit{harvest-and-first-hop} sum time which will be optimized in the second step. The quasi-concavity of function (\ref{eqn:ersr-p5}) can be verified by investigating its second derivative. However, the proof is tedious and we neglect it due to space limit. We will also show the curve of $\mathbb{E}[R_{SR}]$ as proof in the simulation part. The optimal $\kappa$ is given by
solving
\begin{equation}\label{eqn:optk}
\frac{d\mathbb{E}[R_{SR}]}{d\kappa}=0.
\end{equation}
However, there is no close form solution for (\ref{eqn:optk}), we can numerically calculate the optimal $\kappa^*$.

Till now we get the optimal time allocation ratio under a constant sum time $z$. Next, we will find the optimal sum time $z$ that meets the data causality constraint and outage probability constraint. Substitute $\alpha=\frac{\kappa z}{1+\kappa}$ and $\beta=\frac{z}{\kappa +1}$ to (\ref{eqn:p4-1}) and make some simplification we get the equivalent constraint of (\ref{eqn:p4-1}) as
\begin{equation}\label{eqn:cp51}
z\ge\dfrac{\gamma_o(1+\kappa)}{\gamma_o(1+\kappa)+\kappa^2[d(1-d)]^{-\mu}(\kappa d^{\mu}+\gamma_o)^{-1}-\kappa(1-d)^{-\mu}}.
\end{equation}

We denote the right hand of (\ref{eqn:cp51}) as $\widetilde{z_1}(\kappa,  \theta, \gamma_o) $ for convenience of description. Similarly, we introduce $\alpha=\frac{\kappa z}{1+\kappa}, \beta=\frac{z}{\kappa +1}$ to (\ref{eqn:p4-2}) and get its equivalent constraint as
\begin{equation}\label{eqn:p5c2}
\frac{\tau}{\tau-1}\log_2\tau\ge \Psi ,
\end{equation}
where
\begin{equation}\label{eqn:tau}
\tau=(z^{-1}-1)(\kappa^{-1}+1)(1-d)^{-\mu}
\end{equation}
and
\begin{equation}\label{eqn:psi}
\Psi=\frac{(1+\kappa)(1-d)^{\mu}}{d^{\mu}-\kappa}\log_2(\kappa^{-1}d^{\mu}).
\end{equation}

We define function $f(\tau)=\frac{\tau}{\tau-1}\log_2\tau$. Obviously it  is a monotonically increasing function with $\tau$ for $\tau>0$. Assuming $\tau^*=f^{-1}(\Psi)$, we get the equivalent constraint of (\ref{eqn:p4-2}) as
\begin{equation}\label{eqn:p5c21}
z\le \left[\frac{\kappa\tau^*}{(1+\kappa)(1-d)^{\mu}}+1\right]^{-1},
\end{equation}
the right hand of which is denoted as $\widetilde{z_2}(\kappa)$ for easy description.

In conclude, we explain the solving process of problem P4 as following,
\begin{enumerate}
	\item Numerically evaluate equation (\ref{eqn:optk}) to get optimal \textit{harvest ratio} $\kappa^*$.
	\item Solve the following problem
	\begin{eqnarray}\label{eqn:p5}
	\underset{\mathit{z}}{\mathbf{max}} &  & z\dfrac{\log_2(\kappa^{*-1}d^{\mu})}{(1+\kappa^*)(\kappa^{*-1}d^{\mu}-1)} \\
	\mathrm{\mathbf{s.t.}} &  & \widetilde{z_1}\le z\le\widetilde{z_2}.
	\end{eqnarray}
	\item Obtain the optimal time allocation ratio by $\alpha=\frac{\kappa^*z}{1+\kappa^*}$ and $\beta=\frac{z}{1+\kappa^*}$.
\end{enumerate}
\subsection{Simulation Result}
In this part,  we first verify the quasi-concavity of the expected throughput function. Then we investigate the feasible region of the  \textit{harvest-and-first-hop} sum time for varying SIR threshold values. At last, we compare the optimal  throughput performance of direct transmission and DF relay transmission.
In figure.\ref{fig:r_k_df}, the expected throughput of the DF relay system over the \textit{harvest ratio} $\kappa$ is depicted. In this simulation, We set $z=1$, $d=0.5$ and $\mu=2$. The simulation results is given through 10000 independent Rayleigh fading channel realizations, while the analytic results is given by (\ref{eqn:ersr-p5}). It can be found that the two curves match well, thus the quasi-concavity of expected throughput over \textit{harvest ratio} is verified.

\begin{figure}
	\centering
	\includegraphics[width=0.5\linewidth]{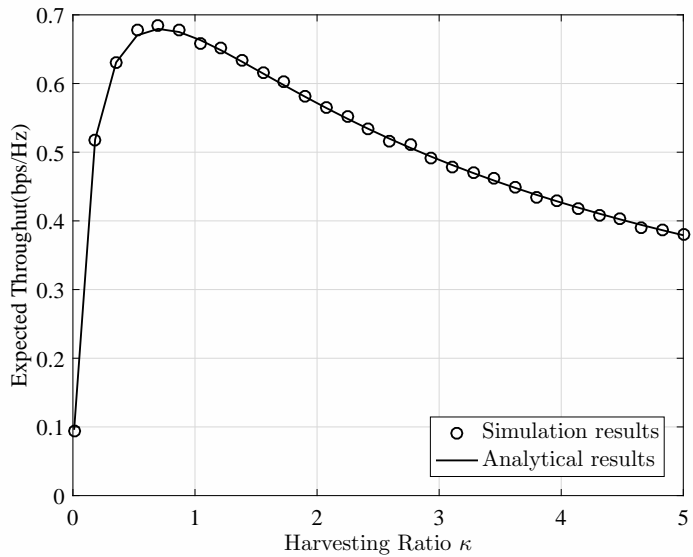}
	\caption{Expected  throughput versus harvesting ratio $\kappa$ for DF relay transmission.}
	\label{fig:r_k_df}
\end{figure}

In figure.\ref{fig:z_sir}, we demonstrate the feasible region of \textit{harvest-and-first-hop} sum time versus outage SIR threshold. In this simulation, we set $d=0$, $\mu=2$ and $\gamma_o$ varies from -20 dB to 0 dB. Since  $\widetilde{z_1}\le z\le\widetilde{z_2}$, we find that the feasible region of $z$ is narrow mainly due to data causality. However, this region for looser outage constraint ($\theta=0.05$) is much wider than that for the tighter one ($\theta=0.02$).

\begin{figure}
	\centering
	\includegraphics[width=0.5\linewidth]{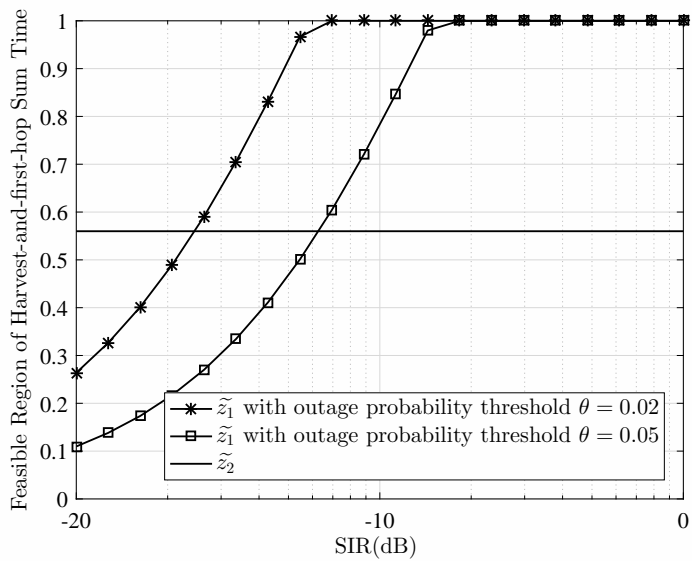}
	\caption{Optimal \textit{harvest-and-first-hop} sum time $z$ versus outage SIR threshold $\gamma_o$ for DF relay transmission.}
	\label{fig:z_sir}
\end{figure}

In figure.\ref{fig:th_d}, we depict the two curves of expected throughput performance over source-relay distance. We set SIR threshold $\gamma_o=-18dB$ and source-relay distance $d$ varies from 0 to 1. The figure shows that when the S-R distance is small, the direct transmission protocol posses better performance, while for larger S-R distance (i.e. $d>0.5$) the cooperative protocol is better. Especially, the performance gain is greatly improved for larger path-loss exponent. This observation coincides the fact that larger path-loss links can gain more benefit from cooperative communications.

\begin{figure}
	\centering
	\includegraphics[width=0.5\linewidth]{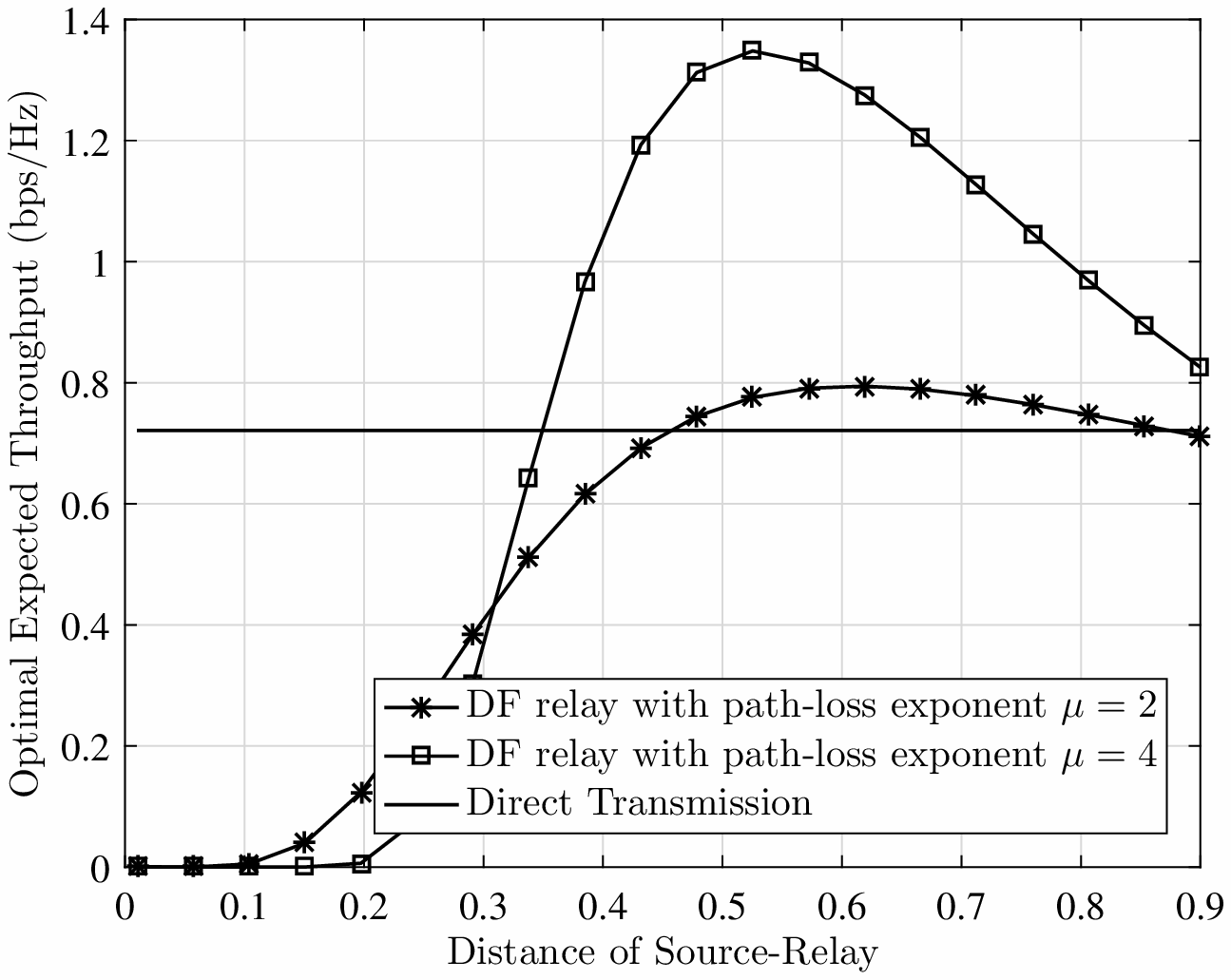}
	\caption{Optimal expected throughput versus source-relay distance for DF relay transmission and direct transmission.}
	\label{fig:th_d}
\end{figure}

\section{Conclusion}

This paper studied the expected throughput optimization for wireless communications powered by non-dedicated sources. We formulate the optimal problems to maximize the expected throughput for direct transmission and DF relay transmission, respectively, subject to outage probability constraint. The optimal {\it harvest ratio} is derived by convex optimization technique. We find that the optimal throughput is irrelevant to the transmitting power of energy source in interference limited environment and upper bounded by a constant $0.5\log_2{e} $(bps/Hz). For most practical  settings, We conclude that  the optimal  {\it harvest ratio} is dominated by outage probability constraint with direct transmission, while by data causality with DF relay transmission .

\appendix
\section*{APPENDIX}
\setcounter{section}{1}
\subsection{Proof of Lemma \ref{lem:cdf}}\label{app:1}
Since the density of ratio of two independent exponential random variables with unit mean is $1/(1+t)^2$  \cite{kadri2014exact}, thus the CDF of $X=k\frac{H_1}{H_2}$ can be obtained as follows,
\begin{align}
F_{X}(x)&=\mathbb{P}(X\le x)\nonumber\\
&=\mathbb{P}(T\le\frac{x}{k})\nonumber\\
&=\int_{0}^{\frac{x}{k}}\frac{1}{(1+t)^2}dt\nonumber\\
&=\frac{x}{k+x}\label{eqn:fgg}
\end{align}
Taking the derivative of (\ref{eqn:fgg}) we finish the proof.
\subsection{Proof of Lemma \ref{lem:mean}}\label{app:2}
\begin{align}
\mathbb{E}[\log(1+x)]&=\int_{0}^{\infty}\log(1+x)f_{X}(x)dx \nonumber\\
&=\int_{0}^{\infty}\log(1+x)\frac{k}{(k+x)^2}dx \nonumber\\
&=\frac{\log(1/k)}{1/k-1} \label{eqn:ec}
\end{align}
\subsection{Proof of Lemma \ref{lem:concdt}}
We have shown that the objective function is continuous in the Lemma 2. Therefore, we investigate the second derivative of the objective function to verify its concavity.

The second derivative of (\ref{eqn:obj2}) is given by
\begin{align}
\frac{d^2\mathbb{E}[R_{DT}]}{d{\alpha}^2}&=\frac{2}{(2\alpha-1)^3}\log_2(\frac{\alpha}{1-\alpha})+\frac{1}{(2\alpha-1)^2\alpha(\alpha-1)\ln2}\notag\\
&\overset{(a)}{=} -\frac{(t+1)^3}{t(1-t)^3\ln 2}(1+2 t\ln t -t^2),\label{eqn:df2}
\end{align}
where (a) comes by letting $t=\frac{\alpha}{1-\alpha}$. As $0<\alpha<1$ so that  $0<t<1$.

We now prove $f(t)=1+2t\ln t-t^2>0$. Suppose $f(x)=1+2x\ln x-x^2, x\in[t,1]$, according to Lagrange Mean Value theorem, there is a number $\xi\in(t,1)$ satisfying the following equation

\begin{align*}
\frac{f(1)-f(t)}{1-t}&= f'(\xi)  \\
&= 2[(1-\xi)+\ln{\xi}] \\
&\le 2[(1-\xi)+(\xi-1)]\\
&=0.
\end{align*}
Then we get $f(1)-f(t)\le 0$ followed by $f(t)\ge f(1)=0$. Integrating this result into (\ref{eqn:df2}),
\begin{displaymath}
\frac{d^2\mathbb{E}[R_{DT}]}{d{\alpha}^2}\le0
\end{displaymath}
and then the concavity of (\ref{eqn:obj2}) is proved.

\bibliographystyle{IEEEtran}
\bibliography{ref}                 

\end{document}